\documentclass{article}

\usepackage[nonatbib, preprint]{neurips_2026}

\usepackage[numbers]{natbib}
\usepackage[utf8]{inputenc} 
\usepackage[T1]{fontenc}    
\usepackage{hyperref}       
\usepackage{url}            
\usepackage{booktabs}       
\usepackage{amsfonts}       
\usepackage{nicefrac}       
\usepackage{microtype}      
\usepackage{xcolor}         

\usepackage{graphicx}
\usepackage{multirow}
\usepackage{svg}
\usepackage{subcaption}
\usepackage{fontawesome}

\makeatletter
\newcommand{\reviewhide}[1]{%
  \if@neuripsfinal
    #1
  \else
    \if@preprint
      #1
    \else
      [removed for review]
    \fi
  \fi
}
\makeatother

\usepackage{amsmath,amsfonts,amsthm,bm}

\newtheorem{theorem}{Theorem}

\def\eqref#1{equation~\ref{#1}}

\def\1{\bm{1}}

\def\ry{{\textnormal{y}}}

\def\rvx{{\mathbf{x}}}

\def\rvz{{\mathbf{z}}}

\def\vw{{\bm{w}}}
\def\vx{{\bm{x}}}
\def\vy{{\bm{y}}}
\def\vz{{\bm{z}}}

\DeclareMathAlphabet{\mathsfit}{\encodingdefault}{\sfdefault}{m}{sl}
\SetMathAlphabet{\mathsfit}{bold}{\encodingdefault}{\sfdefault}{bx}{n}

\def\gD{{\mathcal{D}}}

\def\gF{{\mathcal{F}}}

\def\gH{{\mathcal{H}}}

\def\gL{{\mathcal{L}}}

\def\gX{{\mathcal{X}}}
\def\gY{{\mathcal{Y}}}

\DeclareMathOperator*{\argmax}{arg\,max}
\DeclareMathOperator*{\argmin}{arg\,min}

\title{Non-Linear Strategic Classification Made Practical}

\author{%
  Jack Geary \\
  School of Informatics,\\
  University of Edinburgh,\\
  Edinburgh, United Kingdom\\
  \texttt{jack.geary@ed.ac.uk} \\
  \And
  Boyan Gao \\
  Department of Engineering Science,\\
  University of Oxford,\\ 
  Oxford, United Kingdom \\
  \AND
  Henry Gouk \\
  School of Informatics\\
  University of Edinburgh\\
  Edinburgh, United Kingdom\\
   \texttt{henry.gouk@ed.ac.uk} \\
}

\begin{document}

\maketitle

\newcommand{\keypoint}[1]{\textbf{#1}\qquad}
\theoremstyle{definition}
\newtheorem{definition}{Definition}[section] 

\newtheorem{proposition}{Proposition}[section] 

\begin{abstract}
Algorithmic developments in Strategic Classification have been mostly limited to linear classifiers in settings where the best response has a closed-form solution or can be easily approximated. While some work has explored the role of non-linear classifiers in strategic settings, progress in this direction is impeded by the computational intractability of the strategic behaviour. Addressing this, we present a novel method for approximating the best response by exploiting Lagrangian duality. By reformulating the strategic response as a constrained optimisation problem, we can construct a Lagrangian that is amenable to first order optimisation methods. This approach reproduces closed-form strategic behaviour in linear settings and can be straight-forwardly applied to non-linear settings. We show how the Implicit Function Theorem can be used in conjunction with our proposed response formulation during classifier learning to compute the total gradient of the loss. This connects the classifier parameters directly to the consequent strategic behaviour, yielding a novel training algorithm that can exploit this relationship. Experimental evaluation shows that the resulting models achieve improved strategic accuracy on common machine learning datasets.
\end{abstract}

\begin{center}
    \faGithub \ \reviewhide{\href{github.com/Justme21/strategic-classification-framework}{Experiment Repo}}
\end{center}

\section{Introduction}
Those with a preference for a certain classification, and knowledge of the classifier that is being used, have an incentive to misrepresent their state to receive a favourable outcome \cite{hardt2016}. This dynamic can arise in many sociotechnical situations, spanning from universities deciding which students to enrol as far as institutions determining how to distribute aid and resources in times of need. However the phenomenon is broader than classically considered sociotechnical situations: such strategic behaviour can arise in a broader range of applications than just those that deal with classifying individuals. For example, such dynamics are also a concern in applications of machine learning to protein synthesis \citep{wittmann2025strengthening} and designing benchmarks for generative AI \citep{mcintosh2025inadequacies}. While certain applications are well-served by simple linear models, in many cases deep neural networks have been shown to provide excellent predictive performance.

Strategic classification addresses settings where a \textit{Learner} attempts to classify a population of \textit{Agents} who strategically perturb their representations to optimise a utility known to both players \citep{hardt2016, bruckner2011}. The Learner’s goal is to develop a classifier that remains robust to these perturbations, thereby minimising the impact of manipulation. However, computing the optimal response is a significant challenge, as it is typically discontinuous and non-differentiable in most cases. Consequently, much of the existing literature focuses on linear classifiers, which can allow for closed-form expressions of the best response in some situations \citep{levanon2021, levanon2022, eilat2022}. While this focus on linearity is advantageous for the interpretability required in sociotechnical scenarios, it limits the scope of research and the potential for applying these developments to less socially sensitive settings that could benefit from more flexible models.

To address this limitation we observe that the Agents' objective can be reformulated as a constrained optimisation problem. This admits the utilisation of the Karush Kuhn Tucker (KKT) assumptions to derive the Lagrangian dual corresponding to the problem \citep{kuhn1957}. Since the optima of the Lagrangian dual correspond to the optima of the original optimisation problem, it can be used as a surrogate method for approximating the Agents' best response in general settings. We show that in simple environments our Lagrangian dual-based method for computing the best response better approximates the expected best response than existing gradient-based approaches. Moreover, leveraging this novel method for computing responses to yield an expression for the total derivative of the strategic classification objective. This enables us to explore the utility of directly optimising the standard learning objective, rather than resorting to approaches that optimise for a different solution concept.


In summary, we make the following contributions:
\begin{itemize}
    \item We propose a novel method for approximating the best response to a given classifier based on the Lagrangian Dual  of an alternate formulation of the Agents' objective. This method addresses key weaknesses in current response approximation methods, and can be directly applied in settings with linear or non-linear classifiers.
    \item We observe that the Lagrangian Dual response formulation provides a model for relating the learned parameters to the Agents' behaviour. We propose a novel training algorithm that uses the Implicit Function Theorem to exploit this relationship and demonstrate that the resulting models show improved strategic accuracy on various real world datasets.
    \item We identify a property of some response methods, \textit{reputability}, indicating a method reliably identifies Agent states that are vulnerable to gaming. We prove that the agreement between reputable response methods and the true best response is monotonically related to the number of points the gamed by the response method, allowing for the comparison of different response methods in non-linear settings, without requiring a source of ground-truth gaming behaviour.
\end{itemize}

\section{Related Work}
Building off of the setting established by \cite{bruckner2011, grosshans2013}, Strategic Classification was originally formulated by \citeauthor{hardt2016} \cite{hardt2016}. This models the classifier learning problem as an interaction between a Learner attempting to construct a classifier, and Agents that have perfect awareness of the classifier and can pay perturb their state before being classified. It is assumed that Agents receive positive utility from positive classification, and that the costs the Agents can pay to change their state is known to both the Agents and the Learner. The resulting optimisation problem is generally computationally intractable; in order to empirically evaluate their algorithm, \citeauthor{hardt2016} make the simplifying assumption that the learned classifier is a linear SVM model \cite{hardt2016}. 

Subsequent works in this area have relaxed some of the assumptions made in the original paper; \cite{cohen2024} weaken the assumption that the Agents know the Learner's classifier, \cite{levanon2022} assume that Agents can have a preference for positive or negative classification, \cite{rosenfeld2023b} examines the consequences when the Agents' cost function is unknown to the Learner, and \cite{sundaram2023, braverman2020, geary2025} explore the setting where there is randomness or uncertainty in the classifier learned by the Learner. However, all of these relaxations are considered exclusively in the setting where the classifier being learned is a linear model, and the results are primarily demonstrated on small-scale toy datasets.

\citeauthor{trachtenberg2025} \cite{trachtenberg2025} explicitly explore strategic behaviour for non-linear classifiers. However, their principal motivation is on demonstrating the possible effects of strategic behaviour, and are able to realise this experimentally without having to explicitly compute the best response. \citep{levanon2021} also consider applications of their approach to non-linear classifiers, in the form of recurrent neural network models applied to time series data. However, their solution approach still necessarily relies on model convexity in the inputs, which is not generally true in non-linear settings.

A central motivation for the widespread adoption of the linear classifier in Strategic Classification research is its theoretical simplicity, and that its best response has a closed form definition. This makes it possible to learn strategically robust linear classifiers without actually having to solve the bilevel optimisation problem inherent to the standard problem formulation. This is often used in Strategic Classification literature  (e.g., \citep{levanon2021, levanon2022, eilat2022}). While this facilitates practical experimentation, it limits results to linear models. Some work has been done to weaken this restriction; \citep{levanon2022} in particular, propose a loss function, the Strategic Hinge loss, which can be used when optimising linear models without explicitly computing the best response. However, in general there has been little progress for non-linear models.

Computing the total derivative for optimisation is a commonly used tool in bilevel optimisation \cite{zucchet2022}. The implicit function theorem is widely used to compute the total derivative in these settings \cite{zucchet2022,pedregosa2016hyperparameter, gould2016differentiating}. This approach has been applied in hyperparameter optimization \cite{franceschi2018bilevel, lorraine2020optimizing} meta-learning \cite{rajeswaran2019meta, gao2022loss, gao2022meta}, and differentiable programming \cite{domke2012generic}, providing scalable alternatives to costly unrolled optimization. In contrast with these settings, where it is assumed that the lower level objective is unconstrained, in this work the lower level problem is itself subject to constraints. In fields such as Performative Prediction, recent work has explored using the total derivative in model learning \cite{miller2021, cyffers2024, izzo2021, mofakhami2023, zhong2025}. However, this setting provides no access to the underlying response model; these approaches rely on zeroth order estimation methods to approximate the total derivative. 

\section{Strategic Classification and the Best Response}
Consider a class of models $\gH = \{ \vx \mapsto h_\theta(\vx) \, : \, \theta \in \Theta \}$, parameterised by elements of some space, $\Theta$. The classical classifier learning problem can be formulated as follows; for a feature space, $\gX$, and label space, $\gY = \{-1,1\}$, find a model, $f_\theta \in \gF = \{\vx \mapsto \operatorname{sgn}(h_\theta(\vx)) : h_\theta \in \gH\}$, that achieves high accuracy. If $\gD = \mathcal{P}(\gX \times \gY)$ is the distribution over the features and labels, and $S = \{(\rvx_{i}, \ry_{i})\}_{i=1}^{n}$ is a training set consisting of independent and identically distributed samples from $\mathcal{D}$, then this goal is often addressed by finding the parameters that minimise the empirical risk,
\begin{equation}
    \hat{\theta} = \argmin_{\theta \in \Theta} r(\theta), \qquad r(\theta) = \dfrac{1}{n} \sum_{i=1}^{n}l(\theta, \rvx_{i}, \ry_{i}),
\end{equation}
where $l$ is generally the zero--one loss of $f_\theta$. 

In many realistic classification scenarios $\mathcal{D}$ is not representative of the data that a trained classifier will actually be deployed on \citep{quinonero2022dataset,rosenfeld2023a}. In strategic settings, each data point $(\rvx, \ry)$ may be associated with some Agent that can perturb their representation, $\rvx$. In particular, given that the Agent knows the classifier model, $f_{\theta}$,  they may be motivated to respond to the classifier by manipulating their state in order to obtain a positive classification, irrespective of their true label, $\ry$. This is typically modelled using an idealised best response function \citep{hardt2016},
\begin{equation}
    \Delta^{*}(\vx, \theta) = \argmax_{\vz \in \gX} f_\theta(\vz) - c(\vx, \vz),
    \label{eqn:best_response}
\end{equation}
that trades off receiving a positive classification with some cost, $c(\vx, \vz)$. This cost is assumed to satisfy several natural properties: if there is no manipulation there is no cost, and the cost should be subadditive, $c(\vx,\vz) \leq c(\vx,\vy) + c(\vy,\vz)$ \cite{braverman2020}.

Strategic behaviour can, in practise, significantly undermine the performance of the resulting classifier \citep{levanon2021}. To address this, algorithms have been designed to support training of classifiers that are robust to such misrepresentations \citep{levanon2021, sundaram2023, perdomo2020}. The most common approach to addressing this in the literature the idea of adapting the empirical risk to include this idealised best response.

In Strategic Classification, the objective is to minimise the empirical loss evaluated at the best response, $\Delta(x, \theta)$, which has lead to the Strategic Empirical Risk Minimisation (SERM) approach \citep{sundaram2023},
\begin{align}
    \label{eqn:strategic_classification}
    \tilde{\theta} = \argmin_{\theta \in \Theta} &\frac{1}{n} \sum_{i=1}^n l(\theta, \Delta^{*}(\rvx_i, \theta), \ry_i) \\
    \textup{subject to} &\quad \Delta^{*}(\vx, \theta) = \argmax_{\vz \in \gX} f_\theta(\vz) - c(\vx, \vz). \nonumber
\end{align}
This is a bilevel optimisation problem, where the $\argmin$ over $\theta$ is referred to as the upper level and the $\argmax$ over $\vz$ as the lower level. Observe that computing $\tilde{\theta}$ relies heavily on the ability to solve the lower level problem, which can, for general $\mathcal{F}$ be intractable to solve. As a result, much of the work in the field is focussed on identifying settings and algorithms for which tractable solutions exist \cite{hardt2016, levanon2021, perdomo2020}.

\section{Evaluating Responses}
The solution to the problem posed in Equation \ref{eqn:strategic_classification} relies heavily on the best response, $\Delta^{*}$. While there are limited settings where has a tractable solution, in practise $\Delta^{*}$ generally can't be exactly computed and must be approximated. While previous work has explored methods to approximate $\Delta^{*}$ (e.g., \cite{perdomo2020}), the question of how to evaluate the quality of an approximate response remains open. This is particularly pertinent in the case of responses designed to respond to non-linear classifiers, where the true $\Delta^{*}$ can't be observed.
To address this we propose a measure, \textit{Agreement}, that can be used to quantify the rate of agreement between a response method, and the true best response.

\subsection{Agreement}
Let $\Delta: \mathcal{X}\times \Theta \rightarrow \mathcal{X}$ be an approximate solution to $\Delta^{*}$. A classifier $f_{\theta}, \theta \in \Theta$ is gamed by $\Delta$ at a point $\vx\in \mathcal{X}$ if $f_{\theta}(\Delta(\vx, \theta)) \neq f_{\theta}(x)$. Let $I^{\theta}_{S}(\Delta) = \{\vx \in S| f_{\theta}(\Delta(\vx, \theta)) \neq f_{\theta}(\vx)\}$ be the set of points in $S$ gamed by $\Delta$. For a dataset $S$ and classifier model $f_{\theta}$, the Agreement between two response methods $\Delta_{1}, \Delta_{2}$, $A_{S}^{\theta}(\Delta_{1},\Delta_{2}) \in [0,1]$, is given by the Jaccard Index between the two sets:
\begin{equation}
    A^{\theta}_{S}(\Delta_{1},\Delta_{2}) = \dfrac{|I^{\theta}_{S}(\Delta_{1}) \cap I^{\theta}_{S}(\Delta_{2})|}{|I^{\theta}_{S}(\Delta_{1}) \cup I^{\theta}_{S}(\Delta_{2})|}.
\end{equation}
An Agreement of $0$ implies that the methods did not agree at all, while an agreement of $1$ indicates perfect alignment in what points should be gamed.

$A_{S}^{\theta}(\Delta,\Delta^{*})$ can be seen as an estimate of how well $\Delta$ approximates $\Delta^{*}$ on $S$. In particular, for two response methods, $\Delta_{1}, \Delta_{2}$, if $A_{S}^{\theta}(\Delta_{1}, \Delta^{*})>A_{S}^{\theta}(\Delta_{2},\Delta^{*})$ then $\Delta_{1}$ is a better approximation of $\Delta^{*}$ than $\Delta_{2}$ on $S$.

\subsection{Agreement With an Unobservable Best Response}
In cases when $\Delta^{*}$ can't be computed, to be able to compare response methods, it is necessary to derive a proxy measure to Agreement that does not directly rely on evaluating $\Delta^{*}$. To derive such a measure we will rely on the following definition;
\begin{definition}[Reputability]
    A response method $\Delta$ is said to be \textit{Reputable} if, for any parameter setting $\theta \in \Theta$, $\Delta$ gaming a point $\vx \in \mathcal{X}$ implies that $\vx$ is a point that would be gamed by the best response, $\Delta^{*}$. Specifically:
    \begin{equation}
        f_{\theta}(\Delta(\vx, \theta)) \neq f_{\theta}(\vx) \implies f_{\theta}(\Delta^{*}(\vx, \theta)) \neq f_{\theta}(\vx).
    \end{equation}
\end{definition}
Observe that if $\Delta$ is reputable, then $|I_{S}^{\theta}(\Delta)|$ provides a lower bound on the number of points that are gameable under the unobservable $\Delta^{*}$. As a consequence of this it can be shown that the reputability property is sufficient for $|I_{S}^{\theta}(\Delta)|$ to serve as an adequate proxy for comparing the agreement of two response methods. 
\begin{theorem}
    For two reputable response methods, $\Delta_{1}, \Delta_{2}$, if $|I_{S}^{\theta}(\Delta_{1})|>|I_{S}^{\theta}(\Delta_{2})| \iff A_{S}^{\theta}(\Delta_{1},\Delta^{*})>A_{S}^{\theta}(\Delta_{2}, \Delta^{*})$
    \label{thm:reputable_ranking_property}
\end{theorem}

The proof of Theorem \ref{thm:reputable_ranking_property} is provided in Appendix \ref{sec:proof_of_reputability_theorem}. A direct consequence of this theorem is that, for reputable responses, while it is not possible to determine what ``true" gaming behaviour would be in the setting, it is possible to rank order the responses based on how well they approximate the behaviour by the number of points in $S$ they succeed in gaming.

\section{Computing the Best Response}
\label{sec:computing_best_response}

For a linear model with a squared Euclidean cost,
\begin{equation}
    c(\vx,\vz) = \frac{1}{2\epsilon}||\vx-\vz||_{2}^{2}
    \label{eqn:cost}
\end{equation}
the best response has a closed-form solution and can be computed exactly. In other cases, existing methods simply approximate the best response by optimising differentiable relaxations of the Agent's objective. See Appendix \ref{app:responses} for a more details discussion. We improve upon this by reformulating the original objective in a way that makes it more conducive to gradient-based optimisation.

\subsection{Lagrangian Dual Response}
\label{sec:lagrangian_dual_response}
To address some of the weaknesses highlighted in the previous section, we consider an alternative formulation of the lower level objective. By explicitly addressing the motivation of the lower level objective as being to find a minimum cost response that successfully games the classifier, we can treat it as a constrained optimisation problem. We can write Equation \ref{eqn:best_response} equivalently as,
\begin{align}
        \label{eqn:best_response_bilevel_form}
        \Delta^{*}(\rvx, \theta) = \argmin_{\rvz \in \mathcal{X}} & \quad c(\rvx, \rvz),\\
        \operatorname{subject\,to} &\quad h(\rvz) \geq 0,\nonumber \\
        & \quad c(\rvx, \rvz) \leq 2. \nonumber
\end{align}
In practice, one must constrain $h(\vz) \geq \gamma$ for some small $\gamma>0$ to ensure the inequality holds strictly. Observe that this formulation captures the same constraints as in Equation \ref{eqn:best_response}, but without depending on the interaction between the utility and the cost during the optimisation. 

The constrained optimisation problem proposed in Equation \ref{eqn:best_response_bilevel_form} is amenable to Karush-Kuhn-Tucker (KKT) solution methods \citep{kuhn1957, gordon2012}. Under this formulation we can replace the objective in Equation \ref{eqn:best_response_bilevel_form} with the Lagrangian dual,
\begin{equation}
    \label{eqn:lagrangian}
    \gL(\theta, \vx, \vz, \mu_1, \mu_2) = c(\vx, \vz) + \mu_1(h_\theta(\vz) - \epsilon) + \mu_2(c(\vx, \vz) - 2),
\end{equation}
where $\mu_1, \mu_2 \in \mathbb{R}$ are Lagrange multipliers. This is used to define an optimisation problem,
\begin{align}
    \label{eqn:lagrangian_best_response}
    \Delta^{LD}(\rvx, \theta) = \argmin_{\rvz} \max_{\mu_1, \mu_2} & \quad \mathcal{L}(\theta, \rvx, \rvz, \mu_1, \mu_2)\\
    \operatorname{subject\,to} &\quad \mu_1, \mu_2 \geq 0, \nonumber
\end{align}
that can be solved with projected gradient ascent-descent. The projection function simply clamps the $\mu_1$ and $\mu_2$ Lagrange multipliers to ensure they remain non-negative. The Lagrangian dual solution converges to a locally optimal solution whilst enforcing the constraints. In the case where all local optima are global optima---e.g., if $h_\theta$ and $c$ are convex in $\vz$---this approach is guaranteed to compute the best response.


\subsection{Post-Response Checks}
\label{sec:post_response_checks}
Valid responses must satisfy two invariants, $f_\theta(\Delta(\rvx, \theta))>0$ and $c(\rvx, \rvz)<2$. Both approaches described above may violate these invariants in some instances. The Gradient response will sometimes incur more cost than is rational, given the reward gained from obtaining a positive classification. The Lagrangian Dual may violate one of these constraints if there is no feasible point. In practise, to avoid these issues, we manually check that both invariants are satisfied and return the original $\vx$ if they are not. For the remainder of this work, unless stated otherwise, it can be assumed that the response methods evaluated have these post-response checks applied to them.

\section{Training Strategically Robust Classifiers}
\label{sec:strategic_training}
The solution to the bilevel optimisation problem given in Equation \ref{eqn:strategic_classification} corresponds to a Stackelberg equilibrium. Previous work has shown that, in the setting where the lower objective is as in Equation \ref{eqn:gradient_best_response}, a Nash equilibrium can be approached by alternating between computing the best response to a putative classifier and then fitting a new classifier on the perturbed data \citep{perdomo2020}. It was shown that for sufficiently benevolent data distributions and strategic response objectives, this solution concept can provide approximations of the more desirable Stackelberg equilibria. However it is unclear how good this approximation is in practice. To investigate this, we develop a Strategic Training algorithm that directly optimises the objective in Equation \ref{eqn:strategic_classification} by gradient descent.

The main complication in optimising the objective in Equation \ref{eqn:strategic_classification} is the need to backpropagate through the $\argmax$ in the lower level of the problem. To overcome this, we show how the total derivative of the objective can be computed, facilitating the application of any of the gradient-based optimisation methods typically used for training deep networks \citep{gould2016differentiating, gao2022loss}. The total derivative for a single term of the summation in the strategic empirical risk can be decomposed as
\begin{equation}
\label{eqn:total_derivative}
    \frac{d l}{d\theta} = \underbrace{\frac{\partial l}{\partial \theta}}_{\text{Direct Derivative}} + \underbrace{\frac{\partial l}{\partial \Delta} \frac{\partial \Delta}{\partial \theta}}_{\text{Indirect Derivative}}.
\end{equation}
The direct derivative is computed via the backpropagation process for conventional training of neural networks, but with the strategically manipulated features used in place of the true features. The indirect derivative represents a correction that takes into account how the strategic response depends on $\theta$. We explore two options for computing this term: (i) leveraging implicit gradient methods coupled matrix free linear algebra routines; and (ii) ignoring it completely, thus recovering a version of the REGD approach of \citep{perdomo2020} suitable for the strategic learning setting.

\subsection{Computing the Indirect Derivative}
\label{subsec:computing_the_indirect_derivative}
Computing the indirect derivative is complicated by the presence of the $\argmax$ and the indicator function present in the definition of the best response in Equation \ref{eqn:best_response}. The Lagrangian Dual formulation of the best response given in Equation \ref{eqn:lagrangian_best_response} allows us to circumvent the indicator function, but actually further complicates the $\argmax$ issue by introducing the maximisation over the Lagrange multipliers. We address this problem using so-called implicit gradient techniques that make use of the implicit function theorem and optimality conditions of the lower level problem to yield an expression for the gradient of $\Delta$ with respect to $\theta$.

\begin{theorem}
\label{thm:indirect_derivative}
The indirect derivative is given by
    \begin{align*}
        \frac{\partial}{\partial \Delta}\frac{\partial \Delta}{\partial \theta} = - \begin{bmatrix} \frac{\partial l}{\partial \Delta} & 0 & 0 \end{bmatrix} \begin{bmatrix}
\nabla^2_{\vz\vz} \mathcal{L} & \nabla_{\vz} h_\theta & \nabla_{\vz} c \\
\mu_1 \nabla_{\vz}^\top h_\theta & h_\theta - \epsilon & 0 \\
\mu_2 \nabla_{\vz}^\top c & 0 & c - 2
\end{bmatrix}^{-1} \begin{bmatrix}
\frac{\partial^2 \mathcal{L}}{\partial \theta \partial \vz} \\
\mu_1 \frac{\partial h_\theta}{\partial \theta} \\
0
\end{bmatrix}.
    \end{align*}
\end{theorem}

In the interests of space, the proof of this Theorem is deferred to Appendix \ref{app:proof_of_indirect_derivatives}. This approach to computing the total derivative of the strategic classification objective provides an ideal vehicle for experimentally determining the importance of indirect derivative.


\section{Experimental Results}
We experimentally investigate the utility of our proposed Lagrangian Dual response from two perspectives. First, we evaluate how effectively it can be used to respond to classifiers that have not necessarily been trained to be robust to this specific type of strategic response. Second, we investigate the efficacy of the two strategic training techniques discussed in Section \ref{sec:strategic_training}: directly optimising the total derivative (TGD), or using a version of REGD specialised for strategic settings. Details on model architectures and training procedures are given in Appendix \ref{app:experiment_details}. 

Experiments in this section are performed on a suite of datasets compiled primarily from a subset of the datasets in the TALENT collection (\citep{ye2024}). The set also contains the Housing dataset (\cite{houses2014}), as used in \cite{cyffers2024}, and the Give Me Some Credit dataset (\cite{GiveMeSomeCredit}) as used in \cite{perdomo2020, mofakhami2023}. Further details about the specific datasets used in our experiments can be found in Appendix \ref{sec:dataset_details}.


\subsection{Responding to Linear Models}
\label{sec:responding_to_linear_models}
\begin{figure}[t]
    \centering
    \includegraphics[width=\linewidth]{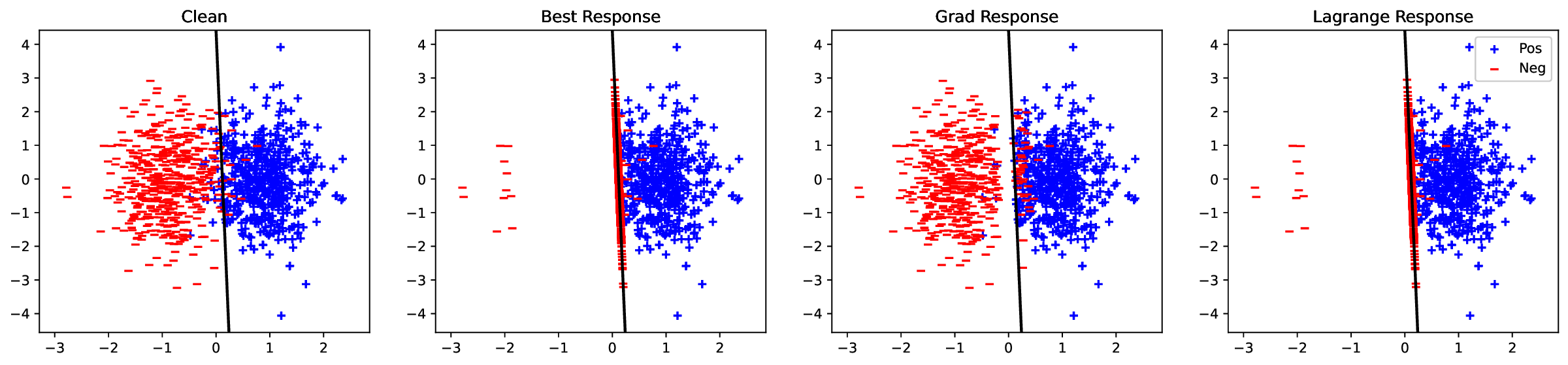}
    \caption{Various responses to a linear classifier trained on the Gaussians dataset (negative points marked with red '-', positive points with blue '+'): (Left) No Response; (Middle-Left) Best Response; (Middle-Right) Gradient Response; (Right) Lagrangian Response}
    \label{fig:compare_responses_balls}
\end{figure}

To demonstrate the effect this response has on the data distribution, we consider the two dimensional Gaussian dataset proposed in \cite{levanon2022}; a simple two dimensional dataset containing points sampled i.i.d from $S=\{(\vx_{i},1)\sim \mathcal{N}(\mu_{1}, 1)\}_{i=1}^{n} \cup \{(\vx_{i},-1)\sim \mathcal{N}(\mu_{-1},1)\}_{i=1}^{n}, \mu_{-1}<\mu_{1}$. We train a linear support vector machine (SVM) on the dataset, without taking any effort to make the classifier strategically robust. We consider three methods for computing strategic responses under the Euclidean cost: the exact closed-form solution, the Gradient response (Equation \ref{eqn:gradient_best_response}), and the Lagrangian Dual response (Equation \ref{eqn:lagrangian_best_response}). For all three approaches we use the squared Euclidean cost. 

Figure \ref{fig:compare_responses_balls} provides a visual comparison of these response methods on the Gaussians dataset. The Left pane shows the unperturbed dataset along with the location of the SVM decision boundary. Comparing with the optimal response (Middle-Left), we observe that the Gradient response (Middle-Right) is susceptible to two types of errors: (i) some points move too far over the decision boundary, meaning they incur more cost than is necessary; and (ii) some points that could game the model are not manipulated. In contrast to this, our Lagrangian Dual respons (Right) is able to exactly replicate the behaviour of the true best response.
\subsection{Responding to Non-Linear Models}
\label{subsec:responding_to_non_linear_models}
\begin{figure}[t]
    \centering
    \includegraphics[width=\linewidth]{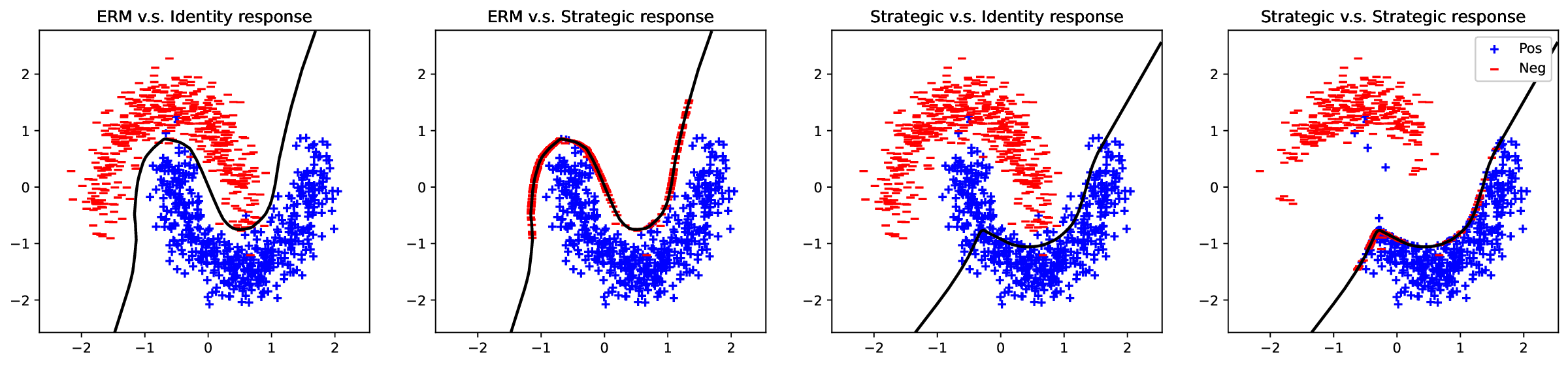}
    \caption{Various responses to an MLP classifier trained on the twin moons dataset (negative points marked with red '-', positive points with blue '+'): (Left) ERM with no response; (Middle-Left) ERM with strategic response; (Middle-Right) REGD-LD with no response; (Right) REGD-LD with strategic response}
    \label{fig:compare_non_linear_responses_twin_moons}
\end{figure}
To qualitatively investigate whether our Lagrangian Dual response is able to respond to non-linear models, we consider two MLP classifiers; one trained with ERM, and another trained using the strategic REGD method \cite{perdomo2020} trained on the twin moons dataset (Figure \ref{fig:compare_strategic_training}).

Figure \ref{fig:compare_non_linear_responses_twin_moons} (Left) visualises the ERM decision boundary and the unperturbed dataset. While this is optimal in absence of strategic behaviour, all of the negative points are close to the decision boundary, and our strategic response is therefore able to move them to lie on the positive side of the decision boundary (Middle-Left). When the model is trained with strategic REGD (Middle-Right), accuracy on the unperturbed data points is poor, but when the strategic behaviour is applied in the Right pane, we see that all the positive points are classified correctly but only a subset of the negatively labelled points successfully game the classifier.
\begin{table*}[t]
\centering
\caption{Proportion (\%) of points gamed for linear and MLP models trained on clean datasets under Gradient ($\Delta^{\text{GD}}$) and Lagrangian ($\Delta^{\text{LD}}$)  response models. Error bars represent 95\% binomial confidence interval. Larger value in each comparison is emphasised.}
\begin{tabular}{lcccc}
\toprule
 & \multicolumn{2}{c}{\textbf{Linear}} & \multicolumn{2}{c}{\textbf{MLP}} \\
\cmidrule(lr){2-3} \cmidrule(lr){4-5}
\textbf{Dataset} & $\Delta^{\text{GD}}$ & $\Delta^{\text{LD}}$ & $\Delta^{\text{GD}}$ & $\Delta^{\text{LD}}$ \\
\midrule
Bank Customer Churn & $3.05 \pm 0.38$ & $\mathbf{24.45 \pm 0.96}$ & $0.40 \pm 0.14$ & $\mathbf{12.80 \pm 0.75}$ \\
default of credit card & $1.27 \pm 0.14$ & $\mathbf{6.02 \pm 0.31}$ & $1.67 \pm 0.17$ & $\mathbf{4.70 \pm 0.27}$ \\
Employee & $0.00 \pm 0.00$ & $\mathbf{0.32 \pm 0.19}$ & $\mathbf{17.29 \pm 1.24}$ & $15.57 \pm 1.19$ \\
GMSC & $3.83 \pm 0.33$ & $\mathbf{24.92 \pm 0.75}$ & $4.97 \pm 0.38$ & $\mathbf{20.85 \pm 0.70}$ \\
Houses & $1.45 \pm 0.19$ & $\mathbf{10.76 \pm 0.48}$ & $1.31 \pm 0.18$ & $\mathbf{9.54 \pm 0.46}$ \\
HR Analytics & $3.76 \pm 0.31$ & $\mathbf{8.77 \pm 0.46}$ & $4.38 \pm 0.33$ & $\mathbf{13.26 \pm 0.55}$ \\
mobile c36 & $34.98 \pm 0.47$ & $\mathbf{37.64 \pm 0.48}$ & $0.00 \pm 0.00$ & $\mathbf{28.75 \pm 0.44}$ \\
statlog & $1.50 \pm 0.86$ & $\mathbf{32.00 \pm 3.30}$ & $1.50 \pm 0.86$ & $\mathbf{25.50 \pm 3.08}$ \\
Water Quality & $0.00 \pm 0.00$ & $\mathbf{16.76 \pm 1.59}$ & $0.73 \pm 0.36$ & $\mathbf{12.75 \pm 1.42}$ \\
\bottomrule
\end{tabular}
\label{table:combined_prop_gamed}
\end{table*}
\subsection{Strategic Responses on Real Data}
Theorem \ref{thm:reputable_ranking_property} presents a means of comparing the relative performance of different response methods in both linear and non-linear settings. However,to make use of this result we assert the following proposition (proof in Appendix \ref{sec:proof_of_prop_post_check_means_reputable});
\begin{proposition}
    Any response, $\Delta$, whose output satisfies to the post-response checks specified in Section \ref{sec:post_response_checks} is a reputable response.
    \label{prop:post_check_means_reputable}
\end{proposition}

Per Theorem \ref{thm:reputable_ranking_property}, reputable responses that game more points are better approximations of $\Delta^{*}$. We compare the performance of $\Delta^{\text{LD}}$ with that of $\Delta^{\text{GD}}$ by measuring the proportions of points they game when responding to linear and MLP models trained with ERM on each of our experiment datasets (Table \ref{table:combined_prop_gamed}). The results show that, in the linear case, $\Delta^{\text{LD}}$ consistently games more points than $\Delta^{\text{GD}}$. Similarly, in the MLP case $\Delta^{\text{LD}}$ consistently outperforms  $\Delta^{\text{GD}}$ in all but one of the datasets. These results indicate that our Lagrangian response, $\Delta^{\text{LD}}$, is reliably a better approximation of $\Delta^{*}$ than the Gradient response.

\subsection{Strategic Training}
\label{sec:regd_vs_tgd}
It remains to be shown that $\Delta^{\text{LD}}$ can improve performance on strategic training. $\Delta^{\text{LD}}$ can be incorporated into training by plugging it directly the strategic REGD training process proposed in \citep{perdomo2020}. An alternative is to train a model with a gradient descent-based algorithm, optimising the total derivative instead of the direct derivative (Section \ref{sec:strategic_training}). In the following results we investigate the consequences of both these approaches.  

\subsubsection{Visualising Strategic Training}
\begin{figure}[t]
    \centering
    \includegraphics[width=\linewidth]{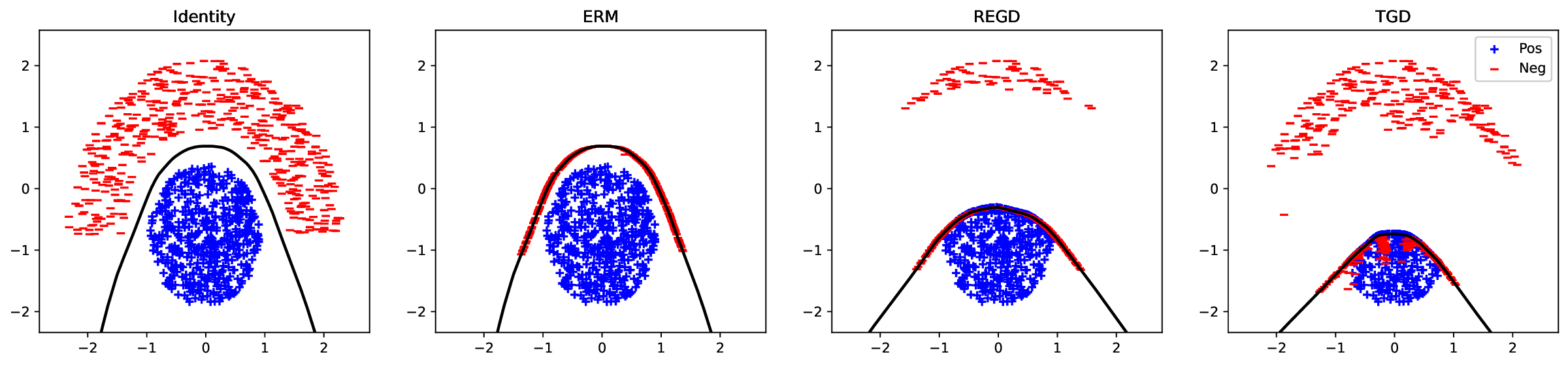}
    \caption{Results of Strategic Training on an MLP on the ball-and-disk dataset (negative points marked with red '-', positive points with blue '+'): (Left) ERM model; (Middle-Left) ERM model with Strategic response; (Middle-Right) REGD model (\cite{perdomo2020}) with Strategic Response; (Right) TGD model (Ours) with Strategic response.}
    \label{fig:compare_strategic_training}
\end{figure}
The impact of using the total gradient (TGD) instead of the direct gradient (REGD) can be observed in Figure \ref{fig:compare_strategic_training}. This depicts the results of training an MLP model on a two dimensional dataset comprised of a positive ball surrounded on one side by an equal number of negative points, with the two clusters separated by a narrow margin. The Left pane depicts ERM training without any consideration of the strategic behaviour, which is extremely vulnerable to strategic behaviour (Middle-Left). The Middle-Right pane presents the results from strategic REGD training. The learned decision boundary does demonstrate improved robustness compared to the ERM result, but is still vulnerable to gaming by the majority of the negative points. In contrast, the model resulting from TGD training (Right) demonstrates considerably greater robustness to gaming behaviour. In particular, we observe that the learner exploits awareness of the gaming potential of positive and negative points, resulting in a decision-boundary with greater distance from the negative points than realised by either other training method.
\subsubsection{Quantitative Evaluation}
\begin{table*}[t]
\centering
\caption{Strategic Accuracies (\%) realised for MLP model trained with different strategic training methods. Strategic accuracies computed with respect to ensemble response evaluated on the respective test datasets. Error bars represent 95\% binomial confidence interval. Best accuracy in each comparison is emphasised.}
\begin{tabular}{lcccccc}
\toprule
\textbf{Dataset} & ERM & REGD & REGD-LD & TGD \\ 
\midrule
Bank Customer Churn & $80.75 \pm 0.88$ & $67.55 \pm 1.05$ & $81.30 \pm 0.87$ & $\mathbf{85.50 \pm 0.79}$ \\
default of credit card & $81.55 \pm 0.50$ & $81.32 \pm 0.50$ & $82.23 \pm 0.49$ & $\mathbf{82.53 \pm 0.49}$ \\
Employee & $74.87 \pm 1.42$ & $76.58 \pm 1.39$ & $76.80 \pm 1.38$ & $\mathbf{76.91 \pm 1.38}$ \\
GMSC & $61.20 \pm 0.84$ & $56.42 \pm 0.86$ & $60.81 \pm 0.84$ & $\mathbf{64.19 \pm 0.83}$ \\
Houses & $93.31 \pm 0.39$ & $93.19 \pm 0.39$ & $95.13 \pm 0.34$ & $\mathbf{96.00 \pm 0.31}$ \\
HR Analytics & $78.84 \pm 0.66$ & $\mathbf{80.25 \pm 0.64}$ & $78.71 \pm 0.66$ & $78.21 \pm 0.67$ \\
mobile c36 & $62.18 \pm 0.48$ & $63.60 \pm 0.47$ & $\mathbf{85.96 \pm 0.34}$ & $64.46 \pm 0.47$ \\
statlog & $\mathbf{71.00 \pm 3.21}$ & $68.00 \pm 3.30$ & $69.50 \pm 3.26$ & $\mathbf{71.00 \pm 3.21}$ \\
Water Quality & $60.11 \pm 2.09$ & $57.56 \pm 2.11$ & $58.29 \pm 2.10$ & $\mathbf{62.30 \pm 2.07}$ \\
\hline
\end{tabular}
\label{table:big_table}
\end{table*}


Strategic accuracy is used to evaluate the effectiveness of the proposed training methods; an MLP with 2 hidden layers and 128 hidden dimensions per hidden layer is trained with each method. The accuracy of the resulting models is evaluated at test time on strategically perturbed data. In order to produce perturbed data that best approximates $\Delta^{*}$, we ensemble over the results of applying $\Delta^{\text{GD}}$ and $\Delta^{\text{LD}}$ to the data with the post-response checks (Section \ref{sec:post_response_checks}). As any point that would be gamed by either response is gamed by the ensemble, this ensemble is still reputable, and is at least as good an approximation to $\Delta^{*}$ as either method individually.  

Table \ref{table:big_table} compares four training processes: ERM with no awareness of strategic behaviour; strategic REGD with the gradient response; strategic REGD with the Lagrangian dual response (REGD-LD); and TGD computed with the Lagrangian Dual response. TGD consistently achieves higher strategic accuracy as compared with the other training methods, demonstrating that incorporating the total gradient into the model training procedure has a significant impact on the resulting model's robustness to strategic behaviour. Comparing the results from REGD with the gradient response and REGD with the Lagrangian dual response, we observe that, while REGD-LD reasonably consistently outperformed the REGD approach, the inclusion of the more accurate response in training does not account for the entirety of the improvement observed in TGD.

\section{Limitations}
While our theoretical results proposed make minimal assumptions that would limit the applicability of the results to more general problems in strategic classifications, our empirical results are confined linear models, and to a single class of non-linear model (MLP). As such they do not speak to the applicability of our proposed approach to more general types of nonlinearities which are not considered in this work. 

\section{Discussion \& Conclusion}
Strategic behaviour can undermine the accuracy and fairness of classifiers. Our intention in this work is to highlight the benefits that can be realised by producing models that are robust to strategic behaviour. However, as is apparent from our results (Figure \ref{fig:compare_non_linear_responses_twin_moons} and Table \ref{fig:compare_strategic_training}), producing models that are robust to strategic behaviour can have the inadvertent consequence of negatively misclassifying some people. These people, who otherwise may not have been inclined to behave strategically, thereby have no option but to attempt to manipulate the classifier, or else risk misclassification. This can put undue financial (or otherwise) burden on people. While some work has been done to explore this consequence in strategic settings (e.g., \citep{milli2019}), an adequate resolution has yet to be found.

In this work we have demonstrated how the Agents' objective in strategic classification can be reformulated as a constrained optimisation problem which is amenable to Lagrangian dual optimisation. A of this approach was the ability to derive the total gradient of the Learner's objective through use of the Implicit Function Theorem. We proposed a novel strategic training algorithm, TGD, which incorporated the total gradient into the training loop, and experimentally demonstrated that models trained with this algorithm consistently outperformed other approaches on strategic accuracy. We further proposed Agreement as a novel metric for evaluating approximations of the best response. We proved that, under minor conditions, the number of points gamed by different approximations could be used as an appropriate proxy to Agreement, preserving rank ordering per this metric.

\bibliographystyle{abbrvnat}
\bibliography{bibliography}

\appendix

\section{Proof of Theorem \ref{thm:reputable_ranking_property}}
\label{sec:proof_of_reputability_theorem}
\begin{proof}
    Without loss of generality, assume every point gamed by $\Delta_{2}$ is also gamed by $\Delta_{1}$: 
    \begin{equation}
        f_{\theta}(\Delta_{2}(\vx, \theta)) \neq f_{\theta}(\vx) \implies f_{\theta}(\Delta_{1}(\vx, \theta)) \neq f_{\theta}(\vx).
    \end{equation}
    Therefore,
    \begin{align}
        &I_{S}^{\theta}(\Delta_{2}) \subseteq I_{S}^{\theta}(\Delta_{1}) \\
        \implies &I_{S}^{\theta}(\Delta_{2})\cap I_{S}^{\theta}(\Delta^{*}) \subseteq I_{S}^{\theta}(\Delta_{1})\cap I_{S}^{\theta}(\Delta^{*}).
    \end{align} 
    $|I_{S}^{\theta}(\Delta_{1})|>|I_{S}^{\theta}(\Delta_{2})| \implies  \exists \vx' \in S \text{ s.t } f_{\theta}(\Delta_{1}(\vx', \theta)) \neq f_{\theta}(\vx')$ and $f_{\theta}(\Delta_{2}(\vx', \theta)) = f_{\theta}(\vx')$. Since $\Delta_{1}$ is reputable, $f_{\theta}(\Delta_{1}(\vx', \theta)) \neq f_{\theta}(\vx') \implies f_{\theta}(\Delta^{*}(\vx', \theta)) \neq f_{\theta}(\vx')$. Therefore $\vx' \in I_{S}^{\theta}(\Delta_{1})\cap I_{S}^{\theta}(\Delta^{*})$ and $\vx' \notin I_{S}^{\theta}(\Delta_{2})\cap I_{S}^{\theta}(\Delta^{*})$. By definition, it follows that $A_{S}^{\theta}(\Delta_{1}, \Delta^{*})>A_{S}^{\theta}(\Delta_{2},\Delta^{*})$.

    To prove equivalence we consider the fact that, for a reputable response $\Delta, I^{\theta}_{S}(\Delta_{1}) \cup I^{\theta}_{S}(\Delta^{*}) = I^{\theta}_{S}(\Delta^{*})$. Therefore
    \begin{align}
        &A_{S}^{\theta}(\Delta_{1}, \Delta^{*})>A_{S}^{\theta}(\Delta_{2},\Delta^{*}) \\
        \implies &|I^{\theta}_{S}(\Delta_{1}) \cap I^{\theta}_{S}(\Delta^{*})|>|I^{\theta}_{S}(\Delta_{2}) \cap I^{\theta}_{S}(\Delta^{*})| \\
        \implies &|I_{S}^{\theta}(\Delta_{1})|>|I_{S}^{\theta}(\Delta_{2})|
    \end{align}
\end{proof}

\section{Best Response Approximations}
\label{app:responses}
\subsection{Linear Response}
\label{app:linear_response}
For a linear model with Euclidean cost, the best response has a closed-form solution and can be computed exactly. As a result, Equation \ref{eqn:strategic_classification} collapses to a single level problem;
\begin{equation}
    \label{eqn:linear-exact}
    \Delta(\vx, \theta) = \begin{cases}
    \vx - \theta \frac{h_\theta(\vx)}{\|\theta\|_2} & \text{if $-2 \leq h_\theta(\vx) < 0$}\\
    \vx & \text{otherwise}.
    \end{cases}
\end{equation}
where, for simplicity we have let $\epsilon=1$. 

\subsection{Gradient Response}
\label{app:gradient_response}
Gradient-based methods are a common approach for solving optimisation problems such as Equation \ref{eqn:strategic_classification}. Such methods involve using gradient ascent to solve the lower level maximisation, and gradient descent to solve the upper level minimisation. \citet{perdomo2020}, propose an iterative gradient-based method that they show can be used to approximate the best response in Strategic Classification problems. However, we note that, since $f$ is a binary classifier, it is not differentiable, and so gradient methods cannot be directly applied to the problem as presented in Equation \ref{eqn:strategic_classification}.

Instead, the response method proposed in \cite{perdomo2020} computes the approximation to a relaxation of the problem in Equation \ref{eqn:strategic_classification}, where the binary classifier, $f_{\theta}$, is replaced in the lower objective with the corresponding $h_{\theta} \in \gH$ (such that $f_{\theta}(x) = sgn(h_{\theta}(x))$), where $\gH$ is a class of differentiable models. This effectively reformulates the best response objective given in Equation \ref{eqn:best_response} as
\begin{equation}
\label{eqn:gradient_best_response}
    \Delta^{GD}(\vx, \theta) = \argmax_{\vz \in \gX} h_\theta(\vz) - c(\vx, \vz).
\end{equation}
As a consequence of this relaxation, one of the constraints on the best response is relaxed; in the best response of Equation \ref{eqn:best_response}, solutions trade off maximising $f_\theta(\vz)$, minimising the cost $c(\vx, \vz)$, and constraining the cost to be less than two ($c(\vx, \vz)\leq2$). This last constraint arises from the interaction between the utility and the cost terms in the optimisation; $f$ is binary and the maximum the utility can be improved by is two. If the cost to realise this exceeds two then the Agent would realise a higher objective by letting $\vz = \vx$. This constraint does not necessarily apply in the case where $f_\theta$ is replaced with $h_\theta$. In particular, we observe that in the linear case with the squared Euclidean cost, the gradient response converges to
\begin{equation}
    \vz = \vx + \epsilon\theta
\end{equation}
which, depending on the value for $\theta$ may wildly over- or under-estimate true gaming behaviour. This is evident from Figure \ref{fig:compare_responses_balls} (centre-right) where the Gradient Response only identifies a subset of the negative points as being vulnerable to gaming and, as such, is a poor approximation of the best response.

\section{Proof of Theorem \ref{thm:indirect_derivative}}
\label{app:proof_of_indirect_derivatives}
\begin{proof}
Denote the concatenation of the primal and dual variables of the Lagrangian Dual response by $\vw = (\vz, \mu_1, \mu_2)$. Consider the residual of this problem,
\begin{align}
    F(\theta, \vx, \vw) &= \begin{bmatrix}
        \nabla_{\vz} \mathcal{L}(\theta, \vx, \vz, \mu_1, \mu_2) \\
        \mu_1 (h_\theta(\vz) - \epsilon) \\
        \mu_2 (c(\vx, \vz) - 2) \\
     \end{bmatrix}.
\end{align}
Let $\vw^\ast(\theta, \vx)$ be the solution to the KKT problem for a given $\theta$ and $\vx$. Then we can define the implicit function,
\begin{equation}
    F(\theta, \vx, \vw^\ast(\theta, \vx)) = 0.
\end{equation}
By the implicit function theorem,
\begin{align}
    0 &= \frac{d F}{d \theta} \\
    0 &= \frac{\partial F}{\partial \theta} + \frac{\partial F}{\partial \vw^\ast} \frac{\partial \vw^\ast}{\partial \theta},
\end{align}
where we have omitted the point where the derivatives are evaluated for readability. Rearranging, we obtain the expression for the indirect derivative,
\begin{equation}
    \frac{\partial \vw^\ast}{\partial \theta} = - \left( \frac{\partial F}{\partial \vw^\ast} \right)^{-1} \frac{\partial F}{\partial \theta}.
\end{equation}
The result follows from computing the derivatives for elements of $F$.
\end{proof}

\section{Experimental Detail}
\label{app:experiment_details}
In this section we provide greater detail about the experiments run in the main paper. All of the experiments were run on python, using standard pytorch packages to train all the models. The code to reproduce all the experimental results will be released after review. All results are from experiments run on a single GeForce RTX 2080 Ti. However, the models used in these experiments are of a size that they could have been run on most standard laptops. 

\subsection{Model Training}
Two classes of models were used for these experiments; Linear models and Multi-Layer Perceptron (MLP) models. The Linear models were standard linear models of the form $W\rvx + \mathbf{b}$, where $W$, the model weights and $\mathbf{b}$, the bias, are learnt during training. The MLP models all had 2 hidden layers, with each hidden layer having 128 dimensions and a final linear layer mapping to a scalar output. A ReLU activation is used between each model layer.

Models were trained on datasets with a 60:20:20 training:validation:test split. Training was run for a maximum of 50 epochs, and early stopping was performed by evaluating the strategic accuracy on the validation set and stopping if no improvement was observed for 10 consecutive epochs. Each epoch of training involved evaluating the strategic response to the current model parameter settings, computing the loss with respect to those strategic values, and then updating the model parameters with respect to that loss. The loss used for our linear results was the hinge loss, while the cross entropy loss was used for training the MLP models. The learning rate used for the training loop was specified per dataset. For more details on this parameter setting see Appendix \ref{app:hyperparameter_optimisation}.

\subsection{Response Computation}
\label{app:response_computation}
Three different response methods were evaluated as part of the experiments in this work; the closed form response to the linear model, the Gradient response ($\Delta^{\text{GD}}$) and the Lagrangian Dual response ($\Delta^{\text{LD}}$). The linear response is parameter free, and was computed according to the definition provided in Equation \ref{eqn:linear-exact}. 

$\Delta^{\text{GD}}$ was computed using standard Gradient ascent-based app on the objective specified in Equation \ref{eqn:gradient_best_response}. $\Delta^{\text{LD}}$ was learned using an Augmented Lagrangian-based approach to optimise the Lagrangian as proposed in Equation \ref{eqn:lagrangian_best_response}. Both of these responses were optimised with the Adam optimiser with a weight decay specified per dataset. Each response was computed by running the iterative optimiser for $1000$ iterations per batch of the training dataset with a learning rate specified per dataset. This number of iterations was deemed sufficient to ensure that both methods had reliably converged and produced plausible gaming behaviour on several datasets. For more details on the setting of the learning rate and the weight decay setting see Appendix \ref{app:hyperparameter_optimisation}.

\subsection{Hyperparameter Optimisation}
\label{app:hyperparameter_optimisation}
The model based optimisation had a single hyperparameter, the learning rate, that was set on a per-dataset basis. The response optimsation had two hyperparameters to be specified on a per-dataset basis; the learning rate and the weight decay. For each hyperparameter we specified a range of plausible values each could attain. Experiments were run for each parameter setting combination and each combination was scored according to the Learner's objective (strategic accuracy) and the Agents' objective (response objective). We then derived the Stackelberg equilibrium treating the Learner as the leader with actions (learning rate, weight decay) and the Agents as the follower with action (response learning rate). The resulting parameters were used in all experimental evaluations. 

\subsection{Implicit Gradient}
In practice, the indirect derivative can be approximated efficiently using a truncated Neumann series, as discussed in \citet{lorraine2020optimizing}. To compute this derivative, we require a solution to the Lagrangian Dual problem given in Equation \ref{eqn:lagrangian_best_response}. In order for training to remain tractable, we aggressively cap the number of iterations used to solve this lower level optimisation problem.

\section{Dataset Details}
\label{sec:dataset_details}
The results in this work made use of datasets from various sources. This section provides the details required to locate or generate the datasets used in this work. Note that, in the process of training our models for our experiments, the dataset values are standardised using the `StandardScaler' as part of the python `sklearn.preprocessing' module. All of the figures presented in this paper, and the results presented, are evaluated on these standardised values. 

\subsection{Toy Datasets}
Figure \ref{fig:compare_responses_balls} is produced from synthetic data generated for the example. The dataset is comprised of a set of positive and negative points sampled i.i.d from a Gaussian distribution with fixed variance. This dataset was sourced from \cite{levanon2022}. For our results we used a variance of $1.5$ for each ball, and and means $10 \pm 2.5$.

Figure \ref{fig:compare_non_linear_responses_twin_moons} is the standard twin-moons dataset as provided by the `sklearn' python package. 

Figure \ref{fig:compare_strategic_training} is realised on a novel dataset; positive points are sampled from a uniform ball with radius $1$. Negative points are them sampled uniformly from a positive half-disk of radius 1 with inner radius 2.  

\subsection{Experiment Datasets}
The datasets used to produce the experimental results in this paper are derived from three sources; the majority of the datasets were sampled from the TALENT dataset collection \cite{ye2024}, while Give Me Some Credit (GMSC, \cite{GiveMeSomeCredit}) and the Houses dataset (\cite{houses2014}) are popular datasets that have previously been used to evaluate results in Strategic or Performative settings.

The GMSC dataset is comprised of historical data relating to borrowers seeking to get a loan from a bank. The motivating objective of the dataset is to predict the probability that a given borrower will default on their loan. \cite{perdomo2020, mofakhami2023} have previously used this dataset to derive their results. However, we note, whereas \cite{perdomo2020} only considered the case where a subset of the dataset features were vulnerable to gaming, the experiments in this paper consider all features equally vulnerable to gaming behaviour.

The Houses dataset is a binarized form of the Housing dataset use in \cite{cyffers2024}. Specific details pertaining to this dataset can be found at \cite{houses2014}. 

The TALENT dataset is a collection of datasets used in the evaluation of Deep Learning models on Tabular data \cite{ye2024}. From the set of classification tasks within this collection, we selected a subset where one could plausibly observe gaming behaviour. The datasets used are:
\begin{itemize}
    \item Bank Customer Churn Dataset
    \item Default of Credit Card Clients
    \item Employee
    \item HR Analytics Job Change of Data Scientists
    \item mobile c36 oversampling
    \item Statlog
    \item Water Quality and Potability
\end{itemize}

\subsection{Choice of Dataset Radius/Cost Function}
In order to elicit gaming behaviour, the $\epsilon$ parameter in the cost (Equation \ref{eqn:cost}) is learned for each dataset. This parameter value is chosen as the value that allowed approximately $50\%$ of the points in the dataset to game a linear classifier trained on the clean data.

\section{Proof of Proposition \ref{prop:post_check_means_reputable}}
\label{sec:proof_of_prop_post_check_means_reputable}
\begin{proof}
    From Section \ref{sec:post_response_checks}, if $f_{\theta}(\Delta(\vx, \theta)) \neq f_{\theta}(\vx)$ after the post-response checks, then $f_{\theta}(\Delta(\vx, \theta))>0$ and $c(\vx, \Delta(\vx, \theta))<2$. Therefore, $\vx$ is a gameable point in response to $\theta$, so $\Delta^{*}$ also games it; $f_{\theta}(\Delta^{*}(\vx, \theta)) \neq f_{\theta}(\vx)$. Hence $f_{\theta}(\Delta(\vx, \theta)) \neq f_{\theta}(\vx) \implies f_{\theta}(\Delta^{*}(\vx, \theta)) \neq f_{\theta}(\vx)$ and the checked $\Delta$ is a reputable response.
\end{proof}

\subsection{Wall Clock Time for REGD vs. REGD-LD vs. TGD}
\begin{figure}[t]
    \centering
    \includegraphics[width=\linewidth]{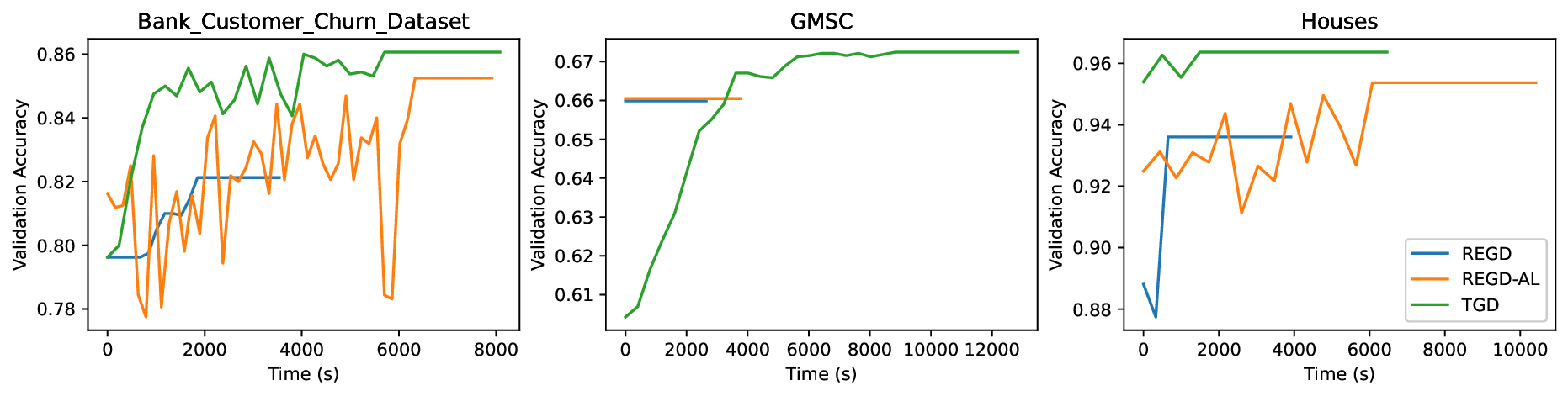}
    \caption{Validation accuracy v.s. runtime (s) for MLP models on subset of TALENT datasets. Models evaluated are; }
    \label{fig:wall_clock_times}
\end{figure}

Figure \ref{fig:wall_clock_times} presents the wall clock times associated with training MLP models on a subset of the TALENT datasets. Times are computed based on experiments run on a single GeForce RTX 2080 Ti. REGD corresponds to the MLP model trained with the strategic REGD with the Gradient response. REGD-AL corresponds to an MLP trained with the strategic REGD using the Lagrangian Dual Response. TGD corresponds to training using Gradient Descent optimising the total gradient instead of the direct gradient.



\end{document}